%% file: main.tex
\documentclass[a4paper,UKenglish,cleveref, autoref, thm-restate]{lipics-v2021}

\usepackage{amsmath,amsthm,amssymb}
\usepackage{booktabs}
\usepackage{xspace}
\usepackage{url}
\usepackage[table]{xcolor}
\usepackage{tikz}
\usepackage{makecell}

\newcommand{\xth}[1]{$#1^{\text{th}}$}

\newcommand{\bftab}{\fontseries{b}\selectfont}

\renewcommand{\qed}{\hfill$\blacksquare$}

\newcommand{\word}{\mathcal{\mathcal{W}}}
\newcommand{\lcp}{\ensuremath{lcp}\xspace}

\newcommand{\BWT}{\ensuremath{\mathrm{BWT}}\xspace}
\newcommand{\RLBWT}{\ensuremath{\mathrm{RLBWT}}\xspace}
\newcommand{\SA}{\ensuremath{\mathrm{SA}}\xspace}

\newcommand{\DA}{\ensuremath{\mathrm{DA}}\xspace}
\newcommand{\LF}{\ensuremath{\mathrm{LF}}\xspace}
\newcommand{\FL}{\ensuremath{\mathrm{FL}}\xspace}
\newcommand{\LCP}{\ensuremath{\mathrm{LCP}}\xspace}

\newcommand{\SPRED}{S_{rank(\pi)}}

\pdfoutput=1 
\hideLIPIcs  

\title{Bounding the Average Move Structure Query for Faster and Smaller RLBWT Permutations}

\titlerunning{Bounded Average Move Structures}

\author{Nathaniel K. {Brown}}{Department of Computer Science, Johns Hopkins University, Baltimore, MD, USA}{nbrown99@jh.edu}{https://orcid.org/0000-0002-6201-2301}{[NSERC PGS-D and NIH grant R01HG011392]}
\author{Ben Langmead}{Department of Computer Science, Johns Hopkins University, Baltimore, MD, USA}{blangme2@jhu.edu}{https://orcid.org/0000-0003-2437-1976}{[NIH grant R01HG011392]}

\authorrunning{N.K. Brown and B. Langmead}

\Copyright{Nathaniel K. Brown and Ben Langmead}

\ccsdesc{Theory of computation~Data compression}

\keywords{Move Structure, Burrows-Wheeler Transform, Permutation}

\category{} 

\supplement{}
\supplementdetails[subcategory={Results Source Code}, cite={}, swhid={}]{Software}{https://github.com/drnatebrown/orbit/releases/tag/bounding_average}

\acknowledgements{The authors thank Ahsan Sanaullah for discussion which motivated this work, and Mohsen Zakeri for helpful insight regarding the length capping proof.}

\nolinenumbers 

\EventEditors{Martin Aum\"{u}ller and Irene Finocchi}
\EventNoEds{2}
\EventLongTitle{24th Symposium on Experimental Algorithms (SEA 2026)}
\EventShortTitle{SEA 2026}
\EventAcronym{SEA}
\EventYear{2026}
\EventDate{June 22--24, 2026}
\EventLocation{Copenhagen, Denmark}
\EventLogo{}
\SeriesVolume{371}
\ArticleNo{24}

\begin{document}

\maketitle

\begin{abstract}
The \textit{move structure} represents permutations with long contiguously permuted intervals in compressed space with optimal query time. They have become an important feature of compressed text indexes using space proportional to the number of Burrows-Wheeler Transform (\BWT) runs, often applied in genomics. This is in thanks not only to theoretical improvements over past approaches, but great cache efficiency and average case query time in practice. This is true even without using the worst case guarantees provided by the interval splitting \textit{balancing} of the original result. In this paper, we show that an even simpler type of splitting, \textit{length capping} by truncating long intervals, bounds the average move structure query time to optimal whilst obtaining a superior construction time than the traditional approach. This also proves constant query time when amortized over a full traversal of a single cycle permutation from an arbitrary starting position.

Such a scheme has surprising benefits both in theory and practice. For a move structure with $r$ runs over a domain $n$, we replace all $O(r \log n)$-bit components to reduce the overall representation by $O(r \log r)$-bits. The worst case query time is also improved to $O(\log \frac{n}{r})$ without balancing. An $O(r)$-time and space construction lets us apply the method to run-length encoded \BWT (\RLBWT) permutations such as $\LF$ and $\phi$ to obtain optimal-time algorithms for \BWT inversion and suffix array (\SA) enumeration in $O(r)$ working space. Finally, we introduce the \texttt{Orbit} library, providing flexible plug and play move structure support, and use it to evaluate our splitting approach. Experiments find length capping construction is faster and uses less memory than balancing, and results in faster move structure queries: up to $\sim 17$ times faster when compared to an unbalanced representation of $\phi$. We also see a space reduction in practice, with at least a $\sim 40\%$ disk size decrease for $\LF$ across large repetitive genomic collections when compared to a balanced/unbalanced move structure.
\end{abstract}


\section{Introduction}

Compressed indexes are an increasingly common tool to answer pattern-matching queries within a compressed-space memory budget.  Such indexes are used in genomics, where users often need to search for strings within large collections of related genome sequences (``pangenomes''). Such collections are highly repetitive, leading to high compression when applying the Burrows-Wheeler Transform (\BWT). As such, many indexes use the BWT or the run-length encoded BWT (\RLBWT) as a foundation for their data structures and algorithms.

The move structure of Nishimoto and Tabei~\cite{nishimoto2021optimal} was recently proposed as a data structure to support ``runny'' permutations consisting of long contiguously permuted intervals in optimal query-time. Since \RLBWT based indexes rely on such permutations, namely the \LF and $\phi$ functions, they can be used to achieve a compressed full-text index that achieves both constant-time queries and $O(r)$-space.  This was a theoretical improvement over the related $r$-index of Gagie, Navarro, and Prezza~\cite{rindex} which uses sparse bitvectors and logarithmic predecessor queries.  Besides the theoretical improvement, other groups noted a major improvement in the move structure's computational efficiency, chiefly due to it's favorable locality of reference~\cite{brown2022rlbwt, zakeri2024movi}.  

However, while excellent worst-case time and space bounds are achievable by ``balancing'' the structure in the manner described by Nishimoto and Tabei~\cite{nishimoto2021optimal}, practical implementations of the move structure have tended to forego balancing and its associated worst-case guarantees, and instead rely on on its practical average-case efficiency. To this end, we developed a simpler splitting procedure called ``length capping'' which lends insight to the practical benefits of the move structure while finally bounding the average query in theory.

Given a permutation $\pi$ over domain $n$ with $r$ contiguously permuted runs (or intervals), we prove that by splitting intervals of length greater than a constant factor of the average length $\frac{n}{r}$, a \textit{length capped} move structure is obtained that achieves constant-time per move query averaged over the $n$ distinct permutation steps in $O(r)$-space. Further, this can be done in $O(r)$-time in comparison to the $O(r \log r)$-time required for balancing~\cite{bertram2024move}.
This also allows us to replace all $O(r \log n)$ bit components of an arbitrary move structure with $O(r \log \frac{n}{r})$ bit representations. Overall, this saves $O(r \log r)$ bits of total space.
A consequence of length capping also bounds the worst case query to $O(\log \frac{n}{r})$ as an alternative to the existing $O(\log r)$-time of an unbalanced representation.\

Further, a length capped move structure provides constant-time queries amortized over $n$ consecutive steps from an arbitrary starting position when the permutation is formed from a single cycle. This is true for \RLBWT permutations $\LF$, $\FL$, $\phi$, and $\phi^{-1}$, allowing us to use amortized complexity alongside the improved construction time to prove optimal $O(n)$-time algorithms for \BWT inversion and suffix array (\SA) enumeration in $O(r)$ working space. This is owed to the unique properties of \RLBWT based permutations which allow efficient construction of an unbalanced move structure before applying length capping.

Finally, we introduce the \texttt{Orbit} library supporting extensive options to represent arbitrary move structures in a flexible plug and play format. Most importantly, this library implements length capping, allowing us to explore its practical benefits. We evaluate its use for \RLBWT permutations \LF and $\phi$ in a genomics setting using collections of human chromosome-19 haplotypes. We find that length capping results in faster average query times in practice, such as a $\sim17$ times faster representation for $\phi$ when compared to an unbalanced representation, and can be used alongside balancing for best results. Its construction when compared to balancing is faster and uses less memory. Further, \texttt{Orbit} achieves the best query time and space of observed methods; in practice, length capping results in at least $\sim40\%$ space reduction for $\LF$ when scaled across large repetitive collections when compared to balanced/unbalanced move structures.

\section{Preliminaries}

\subsection{Notation}

An array $A$ of $|A|=n$ elements is represented as $A[0..n-1]$. $\{A\}$ is the set derived from $A$ of size $|\{A\}|$. For $x \in \mathbb{N}$, let $[0,x]$ represent the array $[0,1,\dots,x]$. A \textit{predecessor query} on a set $B$ with element $x$ is defined inclusively such that $B.pred(x)= \max \{y\in B~|~y\leq x\}$. When dealing with a set of distinct elements, $B.rank(x)$ refers to the index of $B.pred(x)$ in sorted order. A \textit{string} $S[0..n-1]$ is an array of symbols drawn from an ordered alphabet $\Sigma$ of size $\sigma$. 
We use $S[i..j]$ to represent a \textit{substring} $S[i] \dots S[j]$. For strings $S_1$, $S_2$ let $S_1S_2$ represent their concatenation. A \textit{prefix} of $S$ is some substring $S[0..j-1]$, where a \textit{suffix} is some substring $S[i..n-1]$. The length of the \textit{longest common prefix} (\lcp) of $S_1$,$S_2$ is defined as $\lcp(S_1,S_2)=\max \{j~|~S_1[0..j-1] = S_2[0..j-1]\}$. A \textit{text} $T[0..n-1]$ is assumed to be terminated by the special symbol $\$ \notin \Sigma$ of least order so that suffix comparisons are well defined. We use the RAM word model, assuming machines words of size $\word=\Theta(\log n)$ with basic arithmetic and logical bit operations in $O(1)$-time.

\subsection{Suffix Array and Burrows-Wheeler Transform}

The \textit{suffix array} (\SA)~\cite{manber1993suffix} of a text $T[0..n-1]$ is a permutation $\SA[0..n-1]$ of $[0,n-1]$ such that $\SA[i]$ is the starting position of the $i$th lexicographically smallest suffix of $T$. 
Let a collection of documents $\mathcal{D}=\{T_1, T_2, \dots, T_d\}$ be represented by their concatenation $T[0..n-1]=T_1T_2\dots T_d\$$. The \textit{document array} (\DA)~\cite{muthukrishnan2002efficient} is defined as $\DA[0..n-1]$ where $\DA[i]$ stores which document $T[\SA[i]..n]$ begins in. We use $d$ to refer to the number of documents in a collection. The \textit{Burrows-Wheeler Transform} (\BWT)~\cite{bwt} is a permutation $\BWT[0..n-1]$ of $T[0..n-1]$ such that $\BWT[i] = T[\SA[i] - 1]$ if $\SA[i] > 0$, otherwise $\BWT[i]=T[n-1]=\$$. Let $r$ be the number of maximal equal character runs of the \BWT. The \textit{run-length encoded BWT} (\RLBWT) is an array $\RLBWT[0..r-1]$ of tuples where $\RLBWT[i].c$ is the character of the $i$th \BWT run and $\RLBWT.\ell$ its length. 

The \BWT is reversible by using the \textit{last-to-first} (\LF) mapping $\LF(i)$, a permutation over $[0,n-1]$ satisfying $\SA[\LF(i)]=(\SA[i]-1) \mod n$. The \textit{first-to-last} (\FL) mapping is its inverse such that $\FL(i)$ satisfies $\SA[\FL(i)]=(\SA[i]+1) \mod n$. The permutation $\phi$ over $[0,n-1]$ is defined such that $\phi(\SA[i])=\SA[(i-1) \mod n]$, returning the \SA value for the suffix of preceding lexicographic rank. Its inverse $\phi^{-1}$ is defined symmetrically such that $\phi^{-1}(\SA[i])=\SA[(i+1) \mod n]$. The \textit{longest common prefix array} (\LCP) stores the \lcp between lexicographically adjacent suffixes in $T[0..n-1]$, i.e., $\LCP[0..n-1]$ such that $\LCP[0]=0$ and $\LCP[i]=\lcp(T[\SA[i-1]..n-1], T[\SA[i]..n-1])$ for $i>0$. 


\subsection{Move Structure}
For a permutation $\pi$ over $[0,n-1]$, let $S$ be the set of $r$ positions $i$ such that either $i=0$ or $\pi(i-1) \neq \pi(i)-1$. Then given any $i\in[0,n-1]$, we can evaluate $\pi(i)$ using $i'=S.pred(i)$ as 
\begin{equation}
\pi(i)=\pi(i')+(i-i')
\label{eq:pred_map}
\end{equation}
and can do so in $O(r)$-space and $O(\log\log n)$-time by storing $S_\pi = \{\pi(i) \mid \forall i \in S\}$ alongside a compressed sparse bitvector representing $S$ supporting predecessor queries. Intuitively, for any $j \in [0,r-1]$, there exists a \textit{disjoint interval} from $S[j]$ with length $\ell = S[j+1]-S[j]$ (or $\ell=n-S[j]$ if $j=r-1$) such that the \xth{j} interval corresponds to $[S[j],S[j]+\ell)$ within which all positions permute contiguously.

Nishimoto and Tabei's \textit{move structure}~\cite{nishimoto2021optimal} is an alternative formulation in $O(r)$-space avoiding predecessor queries which, assuming we are given the rank of the predecessor for position $i$ in $S$, achieves constant time per step when computing consecutive permutations $\pi^k(i)$ for $k\geq1$ . Assume, for some $i$, that we know the rank $j'$ of its predecessor $i'=S.pred(i)$ in $S$. Then we compute $\pi(i)$ as in (\ref{eq:pred_map}) by accessing $i'=S[j']$ and $\pi(i')=S_\pi[j']$. However, to compute $\pi^2(i)=\pi\left(\pi(i)\right)$ we would need a new predecessor query to find the rank $j''$ of its predecessor $i''=S.pred(\pi(i))$ in $S$. If we also store the position of predecessors of $S_\pi$ in $S$, $\SPRED = \{S.rank(\pi(i)) \mid \forall i \in S\}$, then we could scan down starting from $\SPRED[j']$ in $S$ until we find the largest rank $j''$ such that $S[j''] \leq \pi(i)$. Then, we compute $\pi^2(i)$ as we did before using accesses with $j''$, and continue similarly for any $\pi^k(i)$ with $k > 1$. A single iteration of this procedure is called a \textit{move query}. Figure~\ref{fig:example} shows the relationship between a permutation and its move structure.

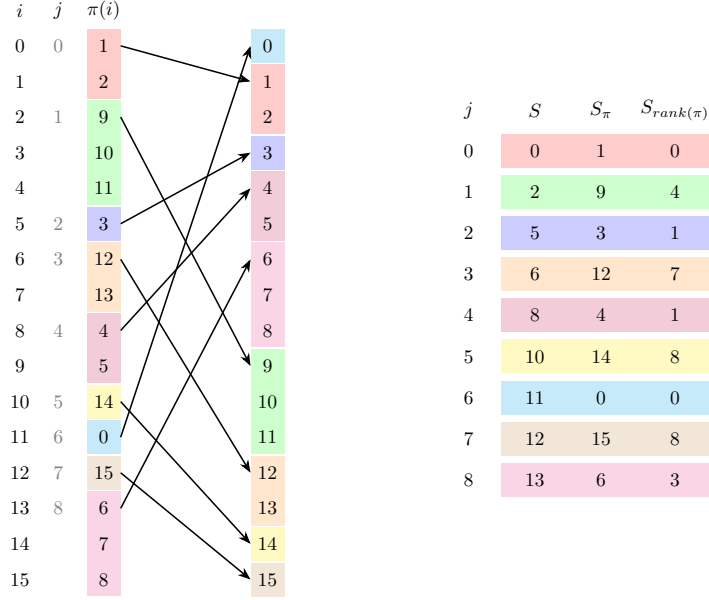
\begin{figure}
    \centering
    \resizebox{0.7\linewidth}{!}{
        \input{Figs/example.tex}
    }
    \caption{For an example permutation $\pi$, shows the corresponding relationship between its contiguously permuted runs and the move structure components $S$, $S_\pi$, and $\SPRED$.}
    \label{fig:example}
\end{figure}

The move query time is proportional to what Zakeri et al.~\cite{zakeri2024movi} call \textit{fast forwards}, the distance from $\SPRED[j']$ to $j''$. As described, there could be $O(r)$ fast forwards in the worst case. Brown, Gagie, and Rossi~\cite{brown2022rlbwt} showed that even over all distinct move queries, $\pi^n(i)$ for any $i$ when formed from a single cycle, the worst case total number of fast forwards is $\Theta(n\cdot r)$ and hence average $\Theta(r)$-time per distinct query. However, Nishimoto and Tabei~\cite{nishimoto2021optimal} proved in their original result that we can \textit{balance} the intervals by arbitrarily splitting some of them. Let $S'\supseteq S$ also contain the positions of the new intervals introduced by their balancing, with $|S'| = r'$. Then the number of fast forwards for any move query is $O(1)$, and $r' \in O(r)$. Brown, Gagie, and Rossi~\cite{brown2022rlbwt} showed how to parameterize balancing for a fixed parameter $\alpha$ into at most $r'=r+\frac{r}{\alpha-1}$ intervals and less than $2\alpha$ fast forwards for any step. Bertram, Fischer, and Nalbach~\cite{bertram2024move} gave a practical $O(r \log r)$-time and $O(r)$-space balancing algorithm for parameter $\alpha$. We call this version with theoretical guarantees a \textit{balanced move structure}, in contrast to the original \textit{unbalanced move structure}.

\begin{theorem}[Nishimoto and Tabei~\cite{nishimoto2021optimal}]
    Given a permutation $\pi$ over $[0,n-1]$, let $S$ be the set of $r$ positions such that either $i=0$ or $\pi(i-1) \neq \pi(i)-1$. We can construct in $O(r \log r)$-time and $O(r)$-space~\cite{bertram2024move} a balanced move structure of size $O(r)$ words, which, given position $i$ and the rank of its predecessor in $S$, computes a move query returning $\pi(i)$ and the rank of $\pi(i)$'s predecessor in $S$ in $O(1)$-time.
\label{thm:nt}
\end{theorem}

\section{Average Case Move Structures}

\subsection{Length Splitting}
In practice, move structures often do not store $S$, $S_\pi$ and $\SPRED$ exactly as described, but in some $O(r)$-space representation of the components which support access for an arbitrary rank $j$. Brown, Gagie, and Rossi~\cite{brown2022rlbwt} suggested a practical improvement by switching to \textit{relative positions}, replacing $S$ with $S_\ell$, where $S_\ell[j]$ is the length of the \xth{j} interval. Using relative positions, a move query now returns the interval and offset within that interval rather than the absolute value; however, representing $S_\ell$ uses less space than representing $S$ naively in practice, since $S_\ell$ is the delta encoding of $S$. Similarly, they noted that, whether using absolute or relative positions, you could replace $S_\pi$ with $S_{\Delta}$, where $S_{\Delta}[j]$ is the offset of $S_\pi[j]$ from its predecessor in $S$.

Brown, Gagie, and Rossi~\cite{brown2022rlbwt} then introduced splitting intervals with respect to a maximum length to practically reduce the number of fast forwards, experimenting with values as a function of the average interval length $\frac{n}{r}$. Bertram, Fischer, and Nalbach~\cite{bertram2024move} (and later, Zakeri et al.~\cite{movi2preprint}) leveraged length splitting using an arbitrary constant to instead reduce the size of representing any component in $S_\ell$ and/or $S_\Delta$. We show that a special type of length splitting, which we call \textit{length capping}, guarantees additional properties with respect to move query time.

\subsection{Length Capping}
Consider \textit{length capping} with a fixed parameter $c$, which splits intervals of a move structure such that the maximum length of any interval is $L=c \cdot \frac{n}{r}$, a constant factor of the average interval length. Let the amortized complexity of a move structure refer to supporting $n$ move queries consecutively from an arbitrary starting position of permutations formed from a single cycle, which is equivalent to performing each of the $n$ distinct move queries. Length splitting in this way gives an average of at most $c+1$ fast forwards across distinct move queries while introducing at most $\frac{r}{c}$ new intervals:

\begin{theorem}
    Given a permutation $\pi$ over $[0,n-1]$, let $S$ be the set of $r$ positions such that either $i=0$ or $\pi(i-1) \neq \pi(i)-1$. Then given an unbalanced move structure, we can construct in $O(r)$-time and space a length capped move structure of size $O(r)$ words such that performing $n$ distinct move queries is $O(n)$-time (average $O(1)$-time per query).
\label{thm:amortized}
\end{theorem}
\begin{proof}
    Let $S'$ be the set of interval start positions after length capping, with $|S'|=r'$, and $S'_\pi=\{\pi(i) \mid \forall i \in S'\}$. Since there can be at most $\frac{r}{c}$ disjoint intervals of length $L$ in a domain of size $n$, it follows that length capping introduces at most $n/L=\frac{n}{c\cdot(n/r)}=\frac{rn}{cn}=\frac{r}{c}$ additional intervals and thus $r'\leq r + \frac{r}{c}\in O(r)$. Thus, length capping can be performed in $O(r)$-time and space.

    Where the \xth{j} input interval has length $\ell=S'[j+1]-S'[j]$ (or $\ell=n - S'[j]$ if $j=r-1$), let its \textit{output interval} be $\left[S'_\pi[j], S'_\pi[j] + \ell\right)$; informally, the contiguously permuted positions of the \xth{j} interval. The maximum number of fast forwards for any query from the \xth{j} interval corresponds to the number of positions in $S'$ which intersect with the output interval of $j$. Because the intervals are disjoint, so are the output intervals, and thus the total number of intersections of positions in $S'$ across all distinct output intervals is at most $r'$.

    Due to length capping, there are at most $k\in[0,L)$ distinct offsets within any interval from its start. Let $W[k]$ be the total number of fast forwards across all distinct move queries whose position is at offset $k$ from the start of its interval. The worst case for $W[k]$ is when the offset $k$ is the last offset in its interval, $k=\ell-1$, and we traverse to the bottom of every output interval while fast forwarding. The total number across all output intervals in this case is exactly the number of intersections of $S'$ with the distinct output intervals. Therefore, $W[k] \leq r'$, and the total number of fast forwards across all $n$ distinct positions is 
    \[
    \sum_{k=0}^{L}W[k]\leq L \cdot r' = c \cdot (\frac{n}{r})\cdot(r+\frac{r}{c})=n\cdot c + n = n(c+1)
    \] 
    and thus given $n$ distinct move queries the total time is $O(n)$ and the average cost of a single query is at most $c+1 \in O(1)$-time.
\end{proof}
  
\begin{corollary}
    A length capped move structure supports performing $n$ consecutive queries from an arbitrary starting position in $O(n)$-time (amortized $O(1)$-time per query) when the permutation is formed from a single cycle.
    \label{cor:amortized}
\end{corollary}

\subsection{Other Results}
\label{sec:other}

A move structure maintaining $O(1)$-time move queries can be represented by $S$ as a sparse bitvector using $r \log \frac{n}{r} + O(r)$-bits, $S_\pi$ in $r \log n$ bits, and $\SPRED$ in $r \log r$ bits. Overall, this representation results in $O(r \log n)$-bits. Although Brown et al.~\cite{brown2024faster} showed that move structures for \LF and \FL can be represented in $O(r \log \frac{n}{r})$-bits, we do not know how to remove the $r \log n$ bit components for general permutations, even when switching $S_\pi$ for $S_\Delta$.

However, after length capping, since any element in $S_\Delta$ is the offset within an interval, it is upper bounded by the maximum length $L \in O(\frac{n}{r})$. Thus, with length splitting, we can support both absolute and relative positions by representing $S_\Delta$ in $O(r \log \frac{n}{r})$-bits. If using relative positions, since any element in $S_\ell$ is also bounded by $L$, we can represent it in the same $O(r \log \frac{n}{r})$-bits as representing $S$ using a sparse bitvector. This may be useful in practice since array access is cache friendly compared to sparse bitvector select queries.

Overall, this results in $O(r \log r + r \log \frac{n}{r})=O(r \log n)$-bits, so the asymptotic bound is unchanged. However, no single component requires $O(r \log n)$-bits. Let $\hat{r}=\Theta(r)$ to suppress constant factors of $r$ introduced by balancing and/or length capping. Then the total space in bits by applying length capping changes from 
$$\hat{r} \log \hat{r} + \hat{r}\log \frac{n}{\hat{r}} + \hat{r} \log n + O(\hat{r}) = 2 \hat{r} \log n + O(\hat{r})$$
to 
$$\hat{r} \log \hat{r} + 2 \hat{r} \log \frac{n}{\hat{r}} + O(\hat{r}) = 2 \hat{r}\log n - \hat{r} \log \hat{r} + O(\hat{r}),$$ 
effectively saving $O(r \log r)$-bits of space. This result can also be applied by length capping alongside balancing.

\begin{theorem}
    Given a permutation $\pi$ over $[0,n-1]$, let $S$ be the set of $r$ positions such that either $i=0$ or $\pi(i-1) \neq \pi(i)-1$. We can construct in $O(r \log r)$-time and $O(r)$-space a length capped and balanced move structure, which, given position $i$ and the rank of its predecessor in $S$, computes a move query returning $\pi(i)$ and the rank of $\pi(i)$'s predecessor in $S$ in $O(1)$-time. Where $\hat{r}=\Theta(r)$, we can do so in $\hat{r} \log \hat{r} + 2\hat{r}\log\frac{n}{\hat{r}} + O(\hat{r})$ bits of space.
\end{theorem}
\begin{proof}
    We find $\hat{r}$ intervals by length capping and then perform balancing; this results in $O(r)$ total intervals and thus the construction bounds are correct and we achieve $O(1)$-time queries by Theorem~\ref{thm:nt}. Consider representing $S$ using a sparse bitvector alongside $S_\Delta$ and $\SPRED$. Since any element in $S_\Delta$ is upper bounded by $L=c \cdot \frac{n}{r}$ it can be represented in $\hat{r} \log \frac{n}{r}$ bits and overall we achieve $\hat{r} \log \hat{r} + 2\hat{r}\log\frac{n}{\hat{r}} + O(\hat{r})$ bits of space.
\end{proof}

Another consequence of length capping is that we can write an alternative bound for the worst case $O(r)$ fast forwards. After Theorem~\ref{thm:amortized}, no single move query requires more than $O(\frac{n}{r})$ fast forwards, since the largest interval is itself of length $O(\frac{n}{r})$. Tatarnikov et al.~\cite{tatarnikov2023moni} noted that you can apply \textit{exponential search} when using absolute positions to compute a move query in time logarithmic to the number of fast forwards. Applying the technique to a length capped move structure gives a new worst case guarantee.

\begin{corollary}
    Using exponential search on a length capped move structure results in worst case $O(\log\frac{n}{r})$-time per move query in $O(r)$-space.
\label{cor:exp}
\end{corollary}

\section{RLBWT Applications}
\subsection{BWT Permutations}
\label{sec:rlbwt}
 
 A major benefit of length capping is the $O(r)$-time construction in comparison to balancing with $O(r \log r)$-time. However, a technicality of Theorem~\ref{thm:amortized} is that it specified being given an unbalanced move structure, required since finding the move structure itself takes superlinear time with respect to $r$. That is, given $S$ and $S_\pi$, finding $\SPRED$ in $O(r)$-space requires computing an $O(r \log r)$-time sorting of $S_\pi$ or computing the required predecessor queries in $O(r \log \log n)$-time using a sparse bitvector. Based on known lower bounds for comparison based sorting and predecessor queries, it would appear that this problem cannot be solved in $O(r)$-time and space for an arbitrary permutation.

 However, various permutations defined by the \RLBWT do benefit from the improved time complexity to construct a length capped move structure. Both $\LF$ and $\phi$ as well as their respective inverses $\FL$ and $\phi^{-1}$ can be represented as $O(r)$-space move structures~\cite{nishimoto2021optimal}, where $r$ is the number of runs in the \BWT. For $\LF$, it easily follows that $\BWT[i]=\BWT[i+1]$ implies $\LF(i+1)=\LF(i)+1$ and thus its permutation has $r$ contiguously permuted runs. For $\phi$, given the properties of \LF, it follows that $$\phi (\SA [i + 1])
= \SA [i]
= \SA [\LF (i)] + 1
= \phi (\SA [\LF (i + 1)]) + 1
= \phi (\SA [i + 1] - 1) + 1$$ and thus there are at most $r$ contiguously permuted runs in the permutation of $\phi$~\cite{rindex}. Further, these \RLBWT based permutations are formed from a single cycle, allowing us to apply the amortized results of Corollary~\ref{cor:amortized}.

The properties of $\RLBWT$ based permutations also permit alternative constructions of their corresponding move structure. For $\LF$ and $\FL$ it is possible to leverage the properties of the $\BWT$ to find $\SPRED$ in $O(r)$-time and space. This is since the \LF permutation is determined by the lexicographic order and individual rank of its \BWT characters: once the \RLBWT has been scanned and $S_\pi$ computed, we can obtain $\SPRED$ by iterating over the intervals by order of their corresponding output as computed by lexicographic rank of its \RLBWT character~\cite{brown2023bwt}. This allows us to apply the result of Theorem~\ref{thm:amortized} to derive the construction time of a length capped move structure for these permutations.

\begin{lemma}
    Given $\RLBWT[0..r-1]$, a length capped move structure for $\LF$ or its inverse $\FL$ can be constructed in $O(r)$-time and space. 
\label{lem:lf}
\end{lemma}

For $\phi$ and $\phi^{-1}$, the best bound in terms of $r$ is still $O(r \log r)$-time and $O(r)$-space, which is the complexity of finding the unbalanced move structure of an arbitrary permutation. However, Sanaullah et al.~\cite{sanaullah2026rlbwt} noted that you could iterate over the text using $\LF$ to sample $\SA$ positions using $O(n)$ consecutive move queries to derive $\SPRED$ in $O(n)$-time and $O(r)$-space. Although $O(r \log r)$-time may often be preferred to $O(n)$-time in practice, the latter is useful when constructing a length capped move structure for the purpose of performing $n$ consecutive queries as described in Corollary~\ref{cor:amortized} so that the final complexity depends only on $n$. Their method assumes a balanced move structure and thus overall they require $O(n+r\log r)$-time when including time to balance. By applying Lemma~\ref{lem:lf}, we achieve new bounds for construction of an unbalanced move structure and a length capped move structure for $\phi$ and $\phi^{-1}$.

\begin{lemma}
    Given $\RLBWT[0..r-1]$ corresponding to $\BWT[0..n-1]$, an unbalanced move structure for $\phi$ or its inverse $\phi^{-1}$ can be constructed in $O(n)$-time and $O(r)$-space.
\label{lem:un_phi}
\end{lemma}
\begin{proof}
    By using the result of Lemma~\ref{lem:lf} in place of a balanced move structure for the original method of Sanaullah et al.~\cite{sanaullah2026rlbwt}.
\end{proof}

\begin{lemma}
    Given $\RLBWT[0..r-1]$, a length capped move structure for $\phi$ or its inverse $\phi^{-1}$ can be constructed in $O(n)$-time and $O(r)$-space.
\label{lem:phi}
\end{lemma}
\begin{proof}
    By applying Theorem~\ref{thm:amortized} to the result of Lemma~\ref{lem:un_phi}.
\end{proof}

\subsection{Optimal-time \BWT Inversion and \SA Enumeration}

To further illustrate the benefit of length capped move structures for scenarios where query time is amortized $O(1)$-time over $n$ consecutive steps, we apply the results of Section~\ref{sec:rlbwt} to known $\RLBWT$ algorithms. The most basic query involving a $\BWT$ is \textit{inversion}, which uses $\LF$ to iterate over every character in $\BWT[0..n-1]$ to return the original text $T[0..n-1]$. Such queries involving traversal using $\LF$ or $\FL$ could be considered iterating in \textit{text order}. However, it is not known how to do this in optimal $O(n)$-time and $O(r)$-space since obtaining $O(1)$-time $\LF$ requires $O(r \log r)$-time balancing (Theorem~\ref{thm:nt}). Using a length capped move structure, we obtain the optimal-time for inversion using only $O(r)$-space in addition to the output $T[0..n-1]$.

\begin{theorem}
    Given $\RLBWT[0..r-1]$ corresponding to $\BWT[0..n-1]$, the original text $T[0..n-1]$ can be inverted in optimal $O(n)$-time and $O(r)$ working space.
\end{theorem}
\begin{proof}
    We build a length capped move structure for $\LF$ (or, alternatively, $\FL$) from Lemma~\ref{lem:lf} in $O(r)$-time and space. Inverting the \BWT requires performing $n$ consecutive move query steps starting from $i=0$ (or run 0 with offset 0). Therefore, we apply Corollary~\ref{cor:amortized} to perform the inversion in $O(n)$-time and $O(r)$ working space. Hence, the entire procedure requires $O(n)$-time and $O(r)$-space in addition to the result $T[0..n-1]$.
\end{proof}

The analogous process for $\phi$ and $\phi^{-1}$ is iterating in \textit{lexicographic order}. For example, enumerating the suffix array ($\SA[0..n-1]$) can be done by performing a full traversal of $\phi^{-1}$ using $n$ consecutive move queries. The document array ($\DA[0..n-1]$) can also be iterated in this fashion so long as the respective document of $\phi^{-1}$ intervals is sampled alongside the move structure; it is trivial to do this sampling in the time and space bounds of constructing a $\phi^{-1}$ move structure itself. This is assuming that documents are concatenated with unique separator characters so that no two $\phi^{-1}$ intervals can contain multiple documents. Therefore, we can also achieve the optimal bounds for listing these arrays in sequential order in $O(r)$-space.

\begin{theorem}
    Given $\RLBWT[0..r-1]$ corresponding to $\BWT[0..n-1]$, the suffix array $\SA[0..n-1]$ can be enumerated in optimal $O(n)$-time and $O(r)$-space.
\end{theorem}
\begin{proof}
    We build a length capped move structure for $\phi^{-1}$ (or, alternatively, $\phi$) from Lemma~\ref{lem:phi} in $O(n)$-time and $O(r)$-space. Enumerating the suffix array $\SA[0..n-1]$ requires $n$ consecutive move queries of $\phi^{-1}$; we apply Corollary~\ref{cor:amortized} to achieve the entire procedure in $O(n)$-time and $O(r)$-space.
\end{proof}

\begin{corollary}
    Given $\RLBWT[0..r-1]$ corresponding to $\BWT[0..n-1]$, the document array $\DA[0..n-1]$ can be enumerated in optimal $O(n)$-time and $O(r)$-space.
\end{corollary}

Finally, we note that in practice the algorithms of this section are useful for streaming and/or external memory approaches. That is, inverting a \BWT in $O(r)$ working space can easily be modified to write the result character by character to disk instead of holding $T[0..n-1]$ in memory during computation. Similarly, enumerating the suffix array or document array is useful for approaches which do not require writing the full result. For example, Shivakumar and Langmead's Mumemto~\cite{shivakumar2025mumemto} tool computes queries such as \textit{multi-maximal unique matches} and \textit{maximal exact matches} with respect to a collection of documents by streaming $\BWT[i]$, $\SA[i]$, $\DA[i]$ and $\LCP[i]$ in order from $i=0$ to $i=n-1$. In this capacity, amortizing over $n$ consecutive move queries using length capping achieves optimal results while being space efficient in practice. Unfortunately, the best known $\LCP$ enumeration algorithm in $O(r)$-space requires a balanced move structure~\cite{sanaullah2026rlbwt} and hence a length capped move structure cannot be applied to improve the result.

\section{Experimental Results}
\subsection{Orbit Library}

We explore length capping through the \texttt{Orbit} move structure library featuring the results of Theorem~\ref{thm:amortized}, available at \url{https://github.com/drnatebrown/orbit}. Inspired by the optimizations of Bertram et al.'s \texttt{Move-r}~\cite{bertram2024move} and Depuydt et al.'s \texttt{b-move}, our move structure is built over a bit packed matrix. That is, we use bitpacking to support the components of the move structure as columns which use the minimum possible fixed bitwidth, while keeping rows contiguous in memory to support cache efficiency. This is essential to length capping, since it gives the maximum benefit to the bounded size of any elements in $S_\ell$ or $S_\Delta$. Most importantly, we accept the length capping parameter $c$ as input to construction.

Beyond a proof of concept for length capping, \texttt{Orbit} is designed as a header only library that can be inserted into existing software. As such, it supports various representations of a move structure. For absolute positions, the components $S$, $S_\Delta$ and $\SPRED$ are stored in a bitpacked matrix using roughly $r \log r + 2 r \log n$ bits (we do not use sparse bitvectors, optimizing instead for speed). Alternatively, relative positions are supported by storing the components $S_\ell$, $S_\Delta$, and $\SPRED$ in roughly $r \log r + 2 r \log \frac{n}{r}$ bits as described in Section~\ref{sec:other}. Move queries support using either a traditional linear scan over fast forwards or exponential search as in Corollary~\ref{cor:exp}.

Within the library are specialized move structures for supporting the \RLBWT based permutations \LF, \FL, $\phi$, and $\phi^{-1}$. These form the basis needed to perform queries such as \textit{count} and \textit{locate} as originally described by Nishimoto and Tabei~\cite{nishimoto2021optimal} when using move structures to construct a full text index for pattern matching. We do not provide an implementation for text indexing, but have designed the implementation such that additional columns can be added to the move structure, such as the \BWT character for \LF. These can either be bitpacked alongside move structure columns for cache efficiency or stored separated in their own dedicated memory. This makes \texttt{Orbit} well suited to be used in various scenarios as a flexible data structure library to construct larger and more intricate applications requiring a move structure.

\subsection{Setup}

To evaluate the efficiency of length capping and its implementation in practice, we evaluate it against approaches which use balancing for the \RLBWT based permutations $\LF$ and $\phi$. To our knowledge, the only known implementation of a generic move structure library supporting balancing is \texttt{Move-r}~\cite{bertram2024move}. However, their approach performs some fixed length splitting that cannot be turned off. For that reason, we also use the \texttt{r-permute} tool\footnote{\url{https://github.com/drnatebrown/r-permute}} to obtain balanced intervals for \LF without any length splitting. Unfortunately, it does not support arbitrary balancing and thus we cannot do the same for $\phi$.

To generate the permutations for \LF and $\phi$, we computed the \RLBWT of $16, 32, 64, 128, 256$ and $1000$ concatenated haplotypes of human chromosome-19 (chr19) used in the experiments of Boucher et al.~\cite{boucher2021phoni}. The specifics of these collections are described in Table~\ref{tab:dataset}. For \LF, we perform the move queries involved in \BWT inversion, and use relative positions (i.e., using $S_\ell$). For $\phi$, we enumerate the suffix array, which requires absolute positions (i.e., using $S$). Both of these processes involve exactly $n$ distinct move queries and thus benefit from length capping. Construction and query time were measured using GNU time on a server with an Intel(R) Xeon(R) Gold 6248R CPU running at 3.00 GHz with 48 cores and 1.5TB DDR4 memory. We report build times and peak memory usage for \texttt{Orbit} across length capping parameters and \texttt{Move-r} across balancing parameters; we do not measure for other modes since they involve interfacing between implementations, whereas these modes are encapsulated by the listed tool.

\begin{table}[!ht]%
    \centering
    \begin{tabular}{l|rrrrrrr}
        \toprule
         & \multicolumn{7}{c}{\textbf{Number of concatenated chr19 sequences ($d$)}} \\
        & 16 & 32 & 64 & 128 & 256 & 512 & 1,000 \\
        \midrule
        $\;n/10^6$ & 946.01 & 1,892.01 & 3,784.01 & 7,568.01 & 15,136.04 & 30,272.08 & 59,125.12 \\
        $\;r/10^4$ & 3,240.02 & 3,282.51 & 3,334.06 & 3,405.40 & 3,561.98 & 3,923.60 & 4,592.68 \\
        $\;n/r$    & 29.20 & 57.64 & 113.50 & 222.24 & 424.93 & 771.54 & 1,287.38 \\
        \bottomrule
    \end{tabular}
    \bigskip
    \caption{Table summarizing the collections of chr19 datasets across various number of documents $d$. Here, $n$ is the size of the \BWT for that collection and $r$ the number of \BWT runs. Each collection is a superset of the previous.}
    \label{tab:dataset}
\end{table}

\subsection{Exploring Interval Splitting}
\label{sec:splits}

We use $16$ copies of chr19 to explore the effect of various parameters for length capping and balancing. Timings in this section are averaged over 5 iterations. To evaluate $\LF$, we compared metrics using one thread for construction and query using:
\begin{itemize}
    \item \texttt{Orbit} on an unbalanced move structure.
    \item \texttt{Orbit} balanced using $\alpha \in \{2,4,8,16\}$ with \texttt{r-permute}. Values larger than $16$ are identical to an unbalanced move structure.
    \item \texttt{Orbit} length capped using $c \in \{0.5,1,2,4,8,16,32\}$.
    \item \texttt{Move-r} balanced using $\alpha \in \{2,4,8,16\}$.We note that \texttt{Move-r} length splits first through binning into preset bitwidths.
    \item \texttt{Orbit} using pairs $(c,\alpha)\in \{(1,2),(2,4),(4,8),...,(32,64)\}$ by length capping first with \texttt{Orbit}, then passing this result to \texttt{Move-r} to balance the length capped intervals. This result is then loaded back into \texttt{Orbit}.
    \item \texttt{Move-r} using the same procedure as above using $(c,\alpha)$ pairs, but this time loaded into \texttt{Move-r}.
\end{itemize}
Note that we chose to use \texttt{Move-r} balancing instead of \texttt{r-permute} for the combined case of length capping and balancing so that we could use the same underlying move structure. This allows a direct comparison of \texttt{Orbit} and \texttt{Move-r} to evaluate the benefit of bitpacking when using both optimizations.

The results of the experiment are shown in Table~\ref{tab:lf}. We observe that balancing alone has little effect compared to an unbalanced move structure; this result has been observed in other studies for \LF, which seems to have good practical query time even unbalanced~\cite{brown2022rlbwt}. However, length capping results in the fastest and smallest approaches, achieving a space reduction of $\sim 46\%$ compared to the unbalanced version. When directly compared to \texttt{Move-r} when using both balancing and length capping, \texttt{Orbit} is faster and smaller, although speed improvements are slight. 

Although it is obvious from our theory why length capping results should be smaller, speed improvements in length capping may also be due to the smaller representation causing more rows of the move structure to fit in cache. Thus, our experiments are too limited to imply that \texttt{Orbit} itself is a faster implementation than \texttt{Move-r}, but does imply the benefit of length capping itself empirically. Construction times are improved over balancing, showing the theoretical improvement gained by length capping. Construction memory is also smaller for reasonable choices of $c$ and $\alpha$, due to the simplicity of splitting by length when compared with balancing. Overall, the best query times combine both length capping and balancing.

\begin{table}[ht!]
\centering
\caption{Results of length capping (with parameter $c$) and balancing (with parameter $\alpha$) for \LF. K denotes thousands. Certain formulations are highlighted in red, yellow, and green to observe comparisons of the methods. Overall best construction time and memory as well as query size and time are bolded. Size refers to disk size of the corresponding data structure. Time is averaged over all move queries.}
\label{tab:lf}
\begin{tabular}{l|rrrrrrrr}
\toprule
\makecell[l]{~\\Structure}  & \makecell[r]{~\\$\alpha$} & \makecell[r]{\\~$c$} & \makecell[r]{Build\\Time (s)} & \makecell[r]{Build\\Mem. (MB)} & \makecell[r]{~\\\# Intervals}    & \makecell[r]{Size\\(MB)}       & \makecell[r]{Avg.\\Time (ns)} &  \\
\midrule
\midrule
\rowcolor{red!20}\texttt{Orbit} & - & -              & 1.11   & 571.03         & 32,400K      & 311.85       & 98.34            &  \\
\midrule
\texttt{Orbit} & 2 &  -  & -  & -          & 33,579K       & 327.39       & 98.79            &  \\
\texttt{Orbit} & 4 & -  & -   & -         & 32,537K      & 313.17       & 99.18            &   \\
\texttt{Orbit} & 8 & -   & -  & -          & 32,430K       & 312.04       & 99.07            &  \\
\rowcolor{red!20} \texttt{Orbit} & 16 & -  & -   & -         & 32,402K      & 311.87       & 98.09            &   \\
\midrule
\texttt{Orbit} & - & 0.5      & 3.14 & 848.23     & 90,473K       & 395.82       & 97.50         &  \\
\texttt{Orbit} & - & 1       & 1.60  & 635.51     & 49,371K       & 222.17      & 95.13         &  \\
\texttt{Orbit} & - & 2        & 1.30 & 573.16    & 36,907K       & 175.30       & 90.36         &  \\
\texttt{Orbit} & - & 4        & 1.16 & 558.00     & 33,856K       & 169.28       & 86.55         &  \\
\rowcolor{green!20} \texttt{Orbit} & - & 8         & \bftab{1.13}  & \bftab{553.23}    & 32,922K       & \bftab{168.73}       & 85.99         &  \\
\texttt{Orbit} & - & 16       & 1.13   & 584.20   & 32,582K       & 175.13       & 86.54         &  \\
\texttt{Orbit} & - & 32       & 1.13  & 583.45   & 32,465K       & 182.61       & 87.61         &  \\
\midrule
\texttt{Move-r} & 2 & -      & 5.55 & 1214.15     & 33,948K      & 305.53       & 89.24            \\
\texttt{Move-r} & 4 & -     & 4.66  & 957.66    & 32,988K       & 296.89       & 89.61           \\
 \texttt{Move-r} & 8 & -     & 4.49 & 785.82     & 32,923K       & 296.31       & 90.00            \\
\rowcolor{yellow!20} \texttt{Move-r} & 16 & -    & 4.48   & 785.76   & 32,922K       & 296.30       & 89.62            \\
\midrule
\texttt{Orbit} & 2 & 1 & - & -           & 50,184K       & 225.83       & 92.60          \\
\texttt{Orbit} & 4 & 2 & - & -           & 37,293K       & 177.14       & 84.55         \\
\texttt{Orbit} & 8 & 4 & - & -           & 34,022K       & 170.10       & 82.95         \\
\texttt{Orbit} & 16 & 8 & - & -           & 32,998K       & 169.12       & 84.33        \\
\rowcolor{green!20} \texttt{Orbit} & 32 & 16 & - & -           & 32,968K       & 168.96       & \bftab{82.84}       \\
\texttt{Orbit} & 64 & 32 & - & -           & 32,951K       & 168.87       & 82.89       \\
\midrule
\texttt{Move-r}& 2 & 1  & - & -           & 50,184K       & 451.66       & 95.68         \\
\texttt{Move-r}& 4 & 2  & - & -           & 37,293K       & 335.64       & 90.45         \\
\texttt{Move-r}& 8 & 4    & - & -           & 34,022K       & 306.19       & 89.28        \\
\texttt{Move-r}& 16 & 8   & - & -            & 32,998K       & 296.99       & 88.89      \\
\rowcolor{yellow!20} \texttt{Move-r}& 32 & 16  & -    & -        & 32,968K       & 296.71       & 90.52      \\
\texttt{Move-r}& 64 & 32  & -  & -           & 32,951K       & 296.56       & 90.14     \\
\bottomrule
\end{tabular}
\end{table}

The equivalent experiment using the $\phi$ permutation was performed with slight changes. First, we cannot obtain a balanced move structure without the implicit length splitting of \texttt{Move-r}, since this was only available for $\LF$. Second, the $\phi$ permutation permits higher $\alpha$ values before stagnation than $\LF$. The results of the experiment are shown in Table~\ref{tab:phi}. We see that unbalanced move structure queries for $\phi$ performs $\sim10$ times slower or worse when compared to splitting approaches. Further, the best query approaches involve some form of balancing, with the very best approaches being a combination of length capping and balancing. This implies that, in contrast to \LF, in practice $\phi$ does not have good query time when left unbalanced, meaning even just the length capping guarantees result in a large performance improvement. 

When directly compared to \texttt{Move-r} when balancing and length capping, query speed improvements of \texttt{Orbit} can be quite small, seemingly correlating with a slight space improvement. However, both the fastest and smallest approaches for queries use \texttt{Orbit} formulations. Note that only \texttt{Orbit} can make proper use of length capping, since its data structures bitpack to fit. Thus, the smaller space improvement could be attributed to requiring absolute positions rather than relative, which limits the benefit of bounding interval lengths. Regardless, length capping is still a clear query speed improvement when compared to an unbalanced move structure. We also again see that construction times for length capping are faster than balancing, with memory also smaller for reasonable choice of $c$ and $\alpha$.

\begin{table}[ht!]
\centering
\caption{Results of length capping (with parameter $c$) and balancing (with parameter $\alpha$) for $\phi$. K denotes thousands. Certain formulations are highlighted in red, yellow, and green to observe comparisons of the methods. Overall best construction time and memory as well as query size and time are bolded. Size refers to disk size of the corresponding data structure. Time is averaged over all move queries.}
\label{tab:phi}
\begin{tabular}{l|rrrrrrrr}
\toprule
\makecell[l]{~\\Structure}  & \makecell[r]{~\\$\alpha$} & \makecell[r]{\\~$c$} & \makecell[r]{Build\\Time (s)} & \makecell[r]{Build\\Mem. (MB)} & \makecell[r]{~\\\# Intervals}    & \makecell[r]{Size\\(MB)}       & \makecell[r]{Avg.\\Time (ns)} &  \\
\midrule
\midrule
\rowcolor{red!20}\texttt{Orbit} & - & -                & 2.15  & 538.63                & 32,400K       & 311.85       & 1,464.09         &  \\
\midrule
\texttt{Orbit} & - & 0.5                       & 6.49 & 1156.55           & 91,621K      & 698.61      & 101.84         &  \\
\texttt{Orbit} & - & 1                         & 5.67  & 768.72          & 60,885K       & 464.25      & 97.23         &  \\
\texttt{Orbit} & - & 2                         & 5.26 & 590.44          & 46,323K       & 359.01      & 96.76         &  \\
\rowcolor{yellow!20}\texttt{Orbit} & - & 4     & 5.10 & 552.52           & 39,228K       & 308.91       & 86.49         &  \\
\texttt{Orbit} & - & 8                       & \bftab{4.92}  & \bftab{535.15}          & 35,728K       & 285.82       & 92.65         &  \\
\texttt{Orbit} & - & 16                        & 5.16  & 560.24          & 33,997K       & 276.23      & 103.59         &  \\
\texttt{Orbit} & - & 32                       & 5.19  & 555.20         & 33,146K       & \bftab{269.31}       & 130.74         &  \\
\midrule
\texttt{Move-r}& 2 & -                     & 21.68   & 2079.67        & 64,490K       & 580.41      & 67.31         &  \\
\texttt{Move-r}& 4 & -                      & 15.49  & 1439.52          & 44,902K       & 404.12       & 59.72         &  \\
\rowcolor{green!20} \texttt{Move-r}& 8& -  & 13.82  & 1258.34          & 39,360K       & 354.24       & 58.73         &  \\
\texttt{Move-r} &16 & -                    & 13.09  & 1187.88          & 37,265K       & 335.38       & 59.48          &  \\
\texttt{Move-r} & 32 & -                    & 12.74 & 1158.03           & 36,304K       & 326.73      & 61.22          &  \\
\texttt{Move-r} & 64 & -                     & 11.95& 1143.90            & 35,854K       & 322.68       & 68.55         & \\
\midrule
\texttt{Orbit} &2&1                & -   & -         & 82,748K       & 641.30       & 90.9            &  \\
\texttt{Orbit} &4&2                & -   & -         & 54,585K       & 423.03       & 82.30         &  \\
\texttt{Orbit} &8&4                & -   & -        & 43,006K       & 338.67       & 69.34         &  \\
\rowcolor{green!20}\texttt{Orbit} & 16 & 8              & -  & -         & 37,577K       & 300.62       & \bftab{53.82}         &  \\
\texttt{Orbit}&32&16              & -    & -      & 36,605K       & 292.84      & 60.40         &  \\
\midrule
\texttt{Move-r}&2&1                 & -   & -         & 82,748K       & 744.74       & 88.55         &  \\
\texttt{Move-r}&4&2                  & -  & -          & 54,585K       & 491.26       & 81.22         &  \\
\texttt{Move-r}&8&4                & -    & -         & 43,006K       & 387.05       & 70.39         &  \\
\rowcolor{green!20}\texttt{Move-r}& 16 & 8   &-              & -            & 37,577K       & 338.20       & 59.16         &  \\
\texttt{Move-r}&32&16               & -   &-         & 36,605K       & 329.44       & 61.36         &  
\end{tabular}
\end{table}

\subsection{Splitting at Scale}

To evaluate splitting approaches for larger collections and as $n/r$ grows, we perform $\LF$ inversion move query steps using \texttt{Orbit}. Overall, we compared an unbalanced move structure, length capping with $c=4$, balancing with $\alpha=8$, and the combination with $c=4$ and $d=8$. These parameters were selected due to their strong performance in both experiments of Section~\ref{sec:splits}.

Results for splitting at scale are shown in Figure~\ref{fig:scale}. Methods which perform length capping are consistently faster than those that do not, with only a slight improvement from balancing when compared to unbalanced. The time per query for length capping and the combination of splitting approaches is comparable with no decisive best approach. With respect to space, length capping approaches are at least $\sim40\%$ smaller than the unbalanced or only balanced move structures. This trend is consistent across varying collections of chr19 sequences.

\begin{figure}[ht!]
    \centering
    \includegraphics[width=\linewidth]{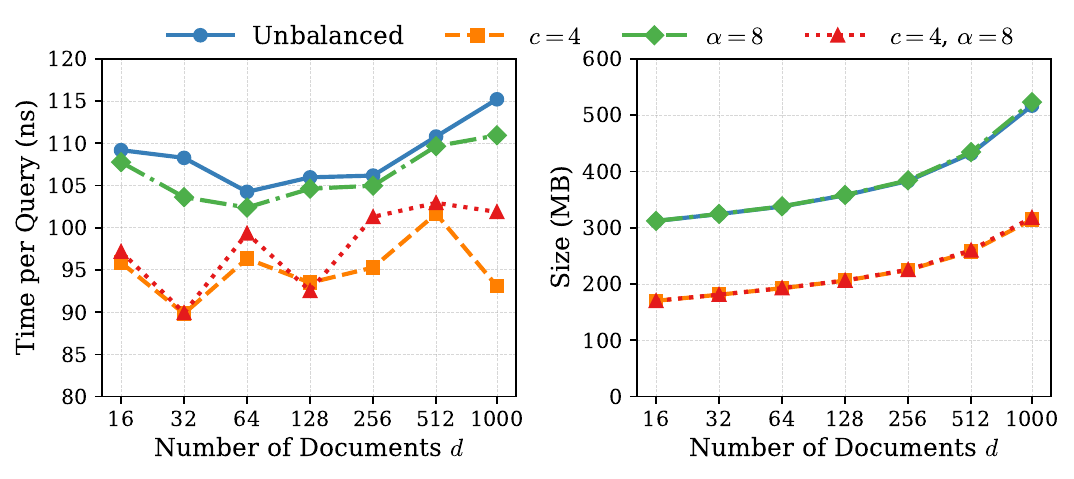}
    \caption{\textbf{Left:} The time per move query in nanoseconds across collections of chr19 haplotypes. \textbf{Right:} The respective disk size in megabytes for the corresponding move structures using \texttt{Orbit}.}
    \label{fig:scale}
\end{figure}

\section{Conclusion}

In this work, we showed an additional splitting regime which achieves average time guarantees with respect to move queries and amortized time guarantees for single cycle permutations while obtaining a $O(r)$-time construction superior to that of balancing. This can be leveraged in many \RLBWT permutations such as \LF and $\phi$ to obtain new optimal-time results in \BWT inversion and \SA enumeration. Further, length capping achieves an $O(r \log r)$-bit space reduction of arbitrary move structures by replacing all $O(r \log n)$-bit components. Finally, we provide a new worst case move query bound for unbalanced move structures of $O(\log \frac{n}{r})$-time.

In practice, length capping shows speed improvements when compared to the unbalanced move structure. It does not always beat balancing alone, such as on $\phi$; however, a combination of balancing and length capping gives the fastest result. Further, theoretical construction time improvements are observed in practice as well as a smaller construction memory footprint. Space improvements are large for $\LF$, although smaller but still observed for $\phi$. This suggests length capping especially useful when using relative positions when the implementations uses bitpacking.

In essence, length capping can be seen as ``fixing'' a distribution whose values sum to at most $n$ even if a single value can be $O(n)$ itself; this is true of both $S_{\ell}$ and $S_\Delta$. It is likely that such an approach could be adopted in other such scenarios. This result could also be combined with Elias-Fano encoding, since lengths store the lower $\log \frac{n}{r}$ bits of elements in $S$ in a cache efficient manner which could be augmented with a bitvector for the higher bits. Further, the unique properties of \RLBWT based permutations allow efficient unbalanced move structure construction. It would seem that most runny permutations would have similar exploitable structures, since understanding their underlying properties appears necessary to derive their function in the first place.

Though usefully applied to streaming/iteration of \RLBWT values, length capping does not yet solve \LCP enumeration. Either an alternative algorithm or improved balancing approaches should be pursued as the \LCP array remains valuable for various advanced match queries. There may also be many other applications for the improved construction bound of length capped move structures with average query time guarantees. To date, there has not been much work on permutations outside an \RLBWT framework for move structures.

The \texttt{Orbit} library proved efficient in practice across various experiments. However, some results show that even if length capping alone is powerful, more desirable results can be achieved by combining with balancing. To promote a powerful feature filled library, the \texttt{Orbit} tool should look at interfacing with the balancing algorithm of \texttt{Move-r} or develop its own implementation.

In summary, length capping is both a useful theoretical and practical result, providing a simpler alternative to balancing whilst still providing average case bounds. This advances applications in \RLBWT permutations and motivates the flexible move structure library \texttt{Orbit}. Since we observed no major downsides to performing such splitting, this result should become commonplace when analyzing or implementing move structures going forward.

\bibliography{refs}


\end{document}

%% file: Figs/example.tex
\usetikzlibrary{matrix,fit,positioning,backgrounds,arrows.meta}

\begin{tikzpicture}[>=Stealth]

\matrix (m) [matrix of nodes, nodes in empty cells,
             column sep=1mm, row sep=0.5mm,
             nodes={minimum width=6mm, minimum height=6mm, anchor=center},
             column 2/.style={text=gray}] 
{
    \textbf{$i$} & \textcolor{black}{$j$} & \textbf{$\pi(i)$} \\ 
    0 & 0 & 1 \\
    1 &   & 2 \\
    2 & 1 & 9 \\
    3 &   & 10 \\
    4 &   & 11 \\
    5 & 2 & 3 \\
    6 & 3 & 12 \\
    7 &   & 13 \\
    8 & 4 & 4 \\
    9 &   & 5 \\
    10 &5 & 14 \\
    11 &6 & 0 \\
    12 &7 & 15 \\
    13 &8 & 6 \\
    14 &  & 7 \\
    15 &  & 8 \\
};

\begin{pgfonlayer}{background}
    \node[fit=(m-2-3)(m-3-3), fill=red!20, inner sep=0pt, outer sep=0pt] {};
    \node[fit=(m-4-3)(m-6-3), fill=green!20, inner sep=0pt, outer sep=0pt] {};
    \node[fit=(m-7-3)(m-7-3), fill=blue!20, inner sep=0pt, outer sep=0pt] {};
    \node[fit=(m-8-3)(m-9-3), fill=orange!20, inner sep=0pt, outer sep=0pt] {};
    \node[fit=(m-10-3)(m-11-3), fill=purple!20, inner sep=0pt, outer sep=0pt] {};
    \node[fit=(m-12-3)(m-12-3), fill=yellow!30, inner sep=0pt, outer sep=0pt] {};
    \node[fit=(m-13-3)(m-13-3), fill=cyan!20, inner sep=0pt, outer sep=0pt] {};
    \node[fit=(m-14-3)(m-14-3), fill=brown!20, inner sep=0pt, outer sep=0pt] {};
    \node[fit=(m-15-3)(m-17-3), fill=magenta!20, inner sep=0pt, outer sep=0pt] {};
\end{pgfonlayer}

\matrix (b) [matrix of nodes, nodes in empty cells,
             right=20mm of m, 
             row sep=0.5mm,
             nodes={minimum width=6mm, minimum height=6mm, anchor=center}] 
{
    \\ 
    0 \\ 1 \\ 2 \\ 3 \\ 4 \\ 5 \\ 6 \\ 7 \\ 8 \\ 9 \\ 10 \\ 11 \\ 12 \\ 13 \\ 14 \\ 15 \\
};

\draw[->, thick] (m-2-3.east) -- (b-3-1.west);   
\draw[->, thick] (m-4-3.east) -- (b-11-1.west);  
\draw[->, thick] (m-7-3.east) -- (b-5-1.west);   
\draw[->, thick] (m-8-3.east) -- (b-14-1.west);  
\draw[->, thick] (m-10-3.east) -- (b-6-1.west);  
\draw[->, thick] (m-12-3.east) -- (b-16-1.west); 
\draw[->, thick] (m-13-3.east) -- (b-2-1.west);  
\draw[->, thick] (m-14-3.east) -- (b-17-1.west);  
\draw[->, thick] (m-15-3.east) -- (b-8-1.west);  

\matrix (s) [matrix of nodes, nodes in empty cells,
             right=25mm of b, 
             row sep=1.5mm,
             nodes={minimum width=12mm, minimum height=6mm, anchor=center}] 
{
    $j$ & \textbf{$S$} & \textbf{$S_\pi$} & \textbf{$\SPRED$} \\ 
    0 & 0 & 1 & 0 \\
    1 & 2 & 9 & 4 \\
    2 & 5 & 3 & 1 \\
    3 & 6 & 12 & 7 \\
    4 & 8 & 4 & 1 \\
    5 & 10 & 14 & 8 \\
    6 & 11 & 0 & 0 \\
    7 & 12 & 15 & 8 \\
    8 & 13 & 6 & 3 \\
};

\begin{pgfonlayer}{background}
    \node[fit=(b-3-1)(b-4-1), fill=red!20, inner sep=0pt, outer sep=0pt] {};
    \node[fit=(b-11-1)(b-13-1), fill=green!20, inner sep=0pt, outer sep=0pt] {};
    \node[fit=(b-5-1)(b-5-1), fill=blue!20, inner sep=0pt, outer sep=0pt] {};
    \node[fit=(b-14-1)(b-15-1), fill=orange!20, inner sep=0pt, outer sep=0pt] {};
    \node[fit=(b-6-1)(b-7-1), fill=purple!20, inner sep=0pt, outer sep=0pt] {};
    \node[fit=(b-16-1)(b-16-1), fill=yellow!30, inner sep=0pt, outer sep=0pt] {};
    \node[fit=(b-2-1)(b-2-1), fill=cyan!20, inner sep=0pt, outer sep=0pt] {};
    \node[fit=(b-17-1)(b-17-1), fill=brown!20, inner sep=0pt, outer sep=0pt] {};
    \node[fit=(b-8-1)(b-10-1), fill=magenta!20, inner sep=0pt, outer sep=0pt] {};
\end{pgfonlayer}

\begin{pgfonlayer}{background}
    \node[fit=(s-2-2)(s-2-4), fill=red!20, inner sep=0pt, outer sep=0pt] {};
    \node[fit=(s-3-2)(s-3-4), fill=green!20, inner sep=0pt, outer sep=0pt] {};
    \node[fit=(s-4-2)(s-4-4), fill=blue!20, inner sep=0pt, outer sep=0pt] {};
    \node[fit=(s-5-2)(s-5-4), fill=orange!20, inner sep=0pt, outer sep=0pt] {};
    \node[fit=(s-6-2)(s-6-4), fill=purple!20, inner sep=0pt, outer sep=0pt] {};
    \node[fit=(s-7-2)(s-7-4), fill=yellow!30, inner sep=0pt, outer sep=0pt] {};
    \node[fit=(s-8-2)(s-8-4), fill=cyan!20, inner sep=0pt, outer sep=0pt] {};
    \node[fit=(s-9-2)(s-9-4), fill=brown!20, inner sep=0pt, outer sep=0pt] {};
    \node[fit=(s-10-2)(s-10-4), fill=magenta!20, inner sep=0pt, outer sep=0pt] {};
\end{pgfonlayer}

\end{tikzpicture}

%% file: refs.bib
@inproceedings{nishimoto2021optimal,
  title={Optimal-Time Queries on {BWT}-Runs Compressed Indexes},
  author={Nishimoto, Takaaki and Tabei, Yasuo},
  booktitle={48th International Colloquium on Automata, Languages, and Programming (ICALP 2021)},
  pages={101--1},
  year={2021},
  organization={Schloss Dagstuhl--Leibniz-Zentrum f{\"u}r Informatik}
}

@article{brown2022rlbwt,
  title={{RLBWT} tricks},
  author={Brown, Nathaniel K and Gagie, Travis and Rossi, Massimiliano},
  journal={LIPIcs: Leibniz international proceedings in informatics},
  volume={233},
  pages={16},
  year={2022}
}

@article{zakeri2024movi,
  title={Movi: a fast and cache-efficient full-text pangenome index},
  author={Zakeri, Mohsen and Brown, Nathaniel K and Ahmed, Omar Y and Gagie, Travis and Langmead, Ben},
  journal={Iscience},
  volume={27},
  number={12},
  year={2024},
  publisher={Elsevier}
}

@inproceedings{bertram2024move,
  title={Move-r: Optimizing the r-index},
  author={Bertram, Nico and Fischer, Johannes and Nalbach, Lukas},
  booktitle={22nd International Symposium on Experimental Algorithms (SEA 2024)},
  pages={1--1},
  year={2024},
  organization={Schloss Dagstuhl--Leibniz-Zentrum f{\"u}r Informatik}
}

@article {movi2preprint,
	author = {Zakeri, Mohsen and Brown, Nathaniel K. and Gagie, Travis and Langmead, Ben},
	title = {Movi 2: Fast and Space-Efficient Queries on Pangenomes},
	year = {2025},
	doi = {10.1101/2025.10.16.682873},
	publisher = {Cold Spring Harbor Laboratory},
	journal = {bioRxiv}
}

@article{manber1993suffix,
  title={Suffix arrays: a new method for on-line string searches},
  author={Manber, Udi and Myers, Gene},
  journal={SIAM Journal on Computing},
  volume={22},
  number={5},
  pages={935--948},
  year={1993},
}

@article{bwt,
  title={A Block-sorting Lossless Data Compression Algorithm},
  author={Burrows, Michael and Wheeler, David J},
  journal={Digital Equipment Corporation},
  numer={124},
  year={1994}
}

@inproceedings{muthukrishnan2002efficient,
  title={Efficient algorithms for document retrieval problems.},
  author={Muthukrishnan, Shanmugavelayutham},
  booktitle={SODA},
  volume={2},
  pages={657--666},
  year={2002},
}

@phdthesis{brown2023bwt,
  title={{BWT}-runs compressed data structures for pan-genomics text indexing},
  author={Brown, Nathaniel},
  year={2023},
  school={Dalhousie University}
}

@article{sanaullah2026rlbwt,
  title={{RLBWT}-Based {LCP} Computation in Compressed Space for Terabase-Scale Pangenome Analysis},
  author={Sanaullah, Ahsan and Brown, Nathaniel K and Shakya, Pramesh and Deegutla, Arun and Naseri, Ardalan and Langmead, Ben and Zhi, Degui and Zhang, Shaojie},
  journal={bioRxiv},
  pages={2026--01},
  year={2026},
  publisher={Cold Spring Harbor Laboratory}
}

@article{shivakumar2025mumemto,
  title={Mumemto: efficient maximal matching across pangenomes},
  author={Shivakumar, Vikram S and Langmead, Ben},
  journal={Genome Biology},
  volume={26},
  number={1},
  pages={169},
  year={2025},
  publisher={Springer}
}

@inproceedings{tatarnikov2023moni,
  title={{MONI} Can Find k-MEMs},
  author={Tatarnikov, Igor and Farahani, Ardavan Shahrabi and Kashgouli, Sana and Gagie, Travis},
  booktitle={34th Annual Symposium on Combinatorial Pattern Matching},
  year={2023}
}

@inproceedings{boucher2021phoni,
  title={{PHONI}: Streamed matching statistics with multi-genome references},
  author={Boucher, Christina and Gagie, Travis and Tomohiro, I and K{\"o}ppl, Dominik and Langmead, Ben and Manzini, Giovanni and Navarro, Gonzalo and Pacheco, Alejandro and Rossi, Massimiliano},
  booktitle={2021 Data Compression Conference (DCC)},
  pages={193--202},
  year={2021},
  organization={IEEE}
}

@article{rindex,
author = {Gagie, Travis and Navarro, Gonzalo and Prezza, Nicola},
title = {Fully Functional Suffix Trees and Optimal Text Searching in {BWT}-Runs Bounded Space},
year = {2020},
issue_date = {February 2020},
publisher = {Association for Computing Machinery},
address = {New York, NY, USA},
volume = {67},
number = {1},
issn = {0004-5411},
url = {https://doi.org/10.1145/3375890},
doi = {10.1145/3375890},
journal = {J. ACM},
month = jan,
articleno = {2},
numpages = {54}
}

@article{brown2024faster,
  title={Faster run-length compressed suffix arrays},
  author={Brown, Nathaniel K and Gagie, Travis and Manzini, Giovanni and Navarro, Gonzalo and Sciortino, Marinella},
  journal={arXiv preprint arXiv:2408.04537},
  year={2024}
}
